\documentclass[onecolumn,amsmath,amssymb,10pt,aps]{revtex4}

\usepackage[sort&compress]{natbib}
\usepackage[utf8x]{inputenc}
\usepackage{amsmath}
\usepackage{amsfonts}
\usepackage{amssymb}
\usepackage{mathrsfs}
\usepackage[]{hyperref}
\usepackage{pgfplots}
\pgfplotsset{compat = newest}
\usepackage{amsthm}

\newtheorem{theorem}{Theorem}
\theoremstyle{plain}

\newtheorem{definition}{Definition}
\newtheorem{example}{Example}

\numberwithin{equation}{section}

\newcommand{\complexes}{\mathbb{C}}

\newcommand{\tr}[1]{\operatorname{tr}{#1}}

\newcommand{\id}[1]{\mathrm{id}_{#1}}
\newcommand{\comm}[2]{[#1,#2]}
\newcommand{\acomm}[2]{\{#1,#2\}}
\newcommand{\dual}{\prime}
\newcommand{\hadj}{*}

\newcommand{\matr}[1]{\mathbb{M}_{#1}(\complexes)}
\newcommand{\matrd}{\matr{d}}

\newcommand{\transpose}{\mathsf{T}}

\DeclareMathOperator{\Tr}{Tr}
\DeclareMathOperator{\der}{\,d}
\newcommand\bigforall{\mbox{\Large $\mathsurround=1pt\forall$}}

\def\<{\langle}
\def\>{\rangle}
\def\oper{{\mathchoice{\rm 1\mskip-4mu l}{\rm 1\mskip-4mu l}
{\rm 1\mskip-4.5mu l}{\rm 1\mskip-5mu l}}}

\begin{document}
\title[Decomposable dynamics on matrix algebras]{Decomposable dynamics on matrix algebras}

\author{Katarzyna Siudzi\'{n}ska\footnote{kasias@umk.pl}}
\affiliation
{Institute of Physics, Faculty of Physics, Astronomy and Informatics, Nicolaus Copernicus University in Toru\'{n}, Grudzi\c{a}dzka 5, 87-100 Toru\'{n}, Poland}

\author{Krzysztof Szczygielski\footnote{krzysztof.szczygielski@ug.edu.pl}}
\affiliation{Institute of Theoretical Physics and Astrophysics, University of Gda\'nsk, Wita Stwosza 57, 80-308 Gda\'nsk, Poland}

\begin{abstract}
We explore a notion of decomposably divisible (D-divisible) quantum evolution families, recently introduced in J.~Phys.~A: Math.~Theor.~\textbf{56} 485202 (2023). Both necessary and sufficient conditions are presented for highly-symmetric qubit and qudit dynamical maps. Through a restructurization of the evolution generators, we encode the decomposable divisibility into the positivity of time-dependent coefficients that multiply generators of D-divisible dynamical maps. This provides an analogy to the CP-divisibility property, which is equivalent to the positivity of decoherence rates that multiply Markovian semigroup generators.
\end{abstract}

\maketitle

\section{Introduction}

Evolution of open quantum systems which are weakly coupled to environment is often approximated by Markovian (i.e.~memoryless) dynamics. However, recent advancements of experimental methods enforce a more nuanced description. Weakly non-Markovian quantum evolution is associated with dynamical maps that are not completely positive (CP) divisible but only positive (P) divisible \cite{Sabrina}. Between P and CP-divisibility, a hierarchy of divisible maps can be distinguished. One example are $k$-divisible dynamical maps, whose definition is related to $k$-positive propagators. Recently, decomposably (D) divisible maps have been established, where propagators are decomposable into a convex combination of completely positive and completely copositive maps \cite{D-div}. Other divisibility properties have also been analyzed in ref. \cite{EB-div} in the form of (eventual) PPT-divisibility and EB-divisibility (positive partial transpose and entanglement breaking, respectively).

An important property of CP-divisible dynamical maps is that they are generated via a Gorini-Kossakowski-Sudarshan-Lindblad (GKSL) generator \cite{GKS,L} with positive decoherence rates. In other words, a Markovian evolution arises from a convex combination of Markovian semigroup generators. In this paper, we generalize this result to D-divisible dynamical maps. That is, our goal is to encode the D-divisibility conditions into the positivity of time-dependent coefficients. Our results include bistochastic qudit evolution that arises from generalized Pauli channels, as well as evolution provided by phase-covariant qubit channels.

The paper is structured as follows. In Section 2, we recall definition and properties of divisible maps, propagators, and generators. Section 3 is dedicated to the simple case of unital qubit maps (Pauli channels), where necessary and sufficient D-divisibility conditions are known. Section 4 generalizes its results to qudit dynamics, where only sufficient conditions can be derived. Section 5 presents an interesting case where D-divisibility conditions are non-linear in the decoherence rates.

\section{D-divisible quantum dynamical maps}

\subsection{Decomposable maps on matrix algebras}

We start with a brief recollection of basic facts regarding decomposable maps. Let $\matrd$ denote an algebra of complex square matrices of size $d \geqslant 2$ and let $\matrd^+$ be the cone of positive semidefinite matrices (for any $a\in\matrd^+$ we simply say $a$ is \emph{positive}, or $a\geqslant 0$). We say that a linear map $\phi : \matrd\to\matrd$ is \emph{positive} if for any $a\geqslant 0$ we have $\phi(a)\geqslant 0$. Map $\phi$ will be further called \emph{completely positive} (CP) if its extension $\id{}\otimes\phi$, acting on tensor product algebra $\matr{d}\otimes\matrd \simeq \matr{d^2}$ via $\phi([a_{ij}]) = [\phi(a_{ij})]$, where $a_{ij}\in\matrd$, is positive. Theorems of Choi and Kraus assert that $\phi$ is CP on $\matrd$ iff its associated \emph{Choi matrix}
\begin{equation}
    C_\phi = \sum_{ij=1}^{d} E_{ij} \otimes \phi(E_{ij})
\end{equation}
is positive in $\matr{d^2}$ ($E_{ij}$ denote matrix units with $1$ in place $(i,j)$ and $0$s elsewhere), which is then true iff $\phi$ admits a \emph{Kraus decomposition}
\begin{equation}
    \phi(a) = \sum_{j} X_j a X_{j}^{\hadj}
\end{equation}
for some (nonunique) family of matrices $\{X_j\}\in\matrd$ (called \emph{Kraus operators}).

Let now $\theta$ denote the usual \emph{transposition} on $\matrd$, i.e.~$\theta(a) = a^\transpose$. A linear map $\xi$ is called \emph{completely copositive} (coCP) if $\theta\circ\xi$ is CP. For every coCP map $\xi$ there exists a CP map $\psi$ s.t.~$\xi = \theta\circ\psi$. Notions of both complete positivity and complete copositivity allow to define yet another, important class of positive maps, so-called \emph{decomposable maps}:
\begin{definition}\label{def:DecomposableMap}
    We say that a linear map $\varphi : \matrd\to\matrd$ is \emph{decomposable} if it can be expressed as a convex combination of CP and coCP map, i.e.~when there exist two CP maps $\phi,\psi$ s.t.
\begin{equation}
    \varphi = \phi + \theta\circ\psi .
\end{equation}
\end{definition}
All decomposable maps are positive. The converse is true in case of endomorphisms on $\matr{2}$ (and also for positive maps $\matr{2}\to\matr{3}$ and vice versa \cite{Woronowicz1976,Stoermer1963}), however fails already for $\matr{3}$ where examples of nondecomposable maps may be constructed. In general case, decomposable maps form a proper convex subcone of cone of all positive maps on $\matrd$. Despite an enormous effort of many mathematicians throughout last few decades, a general characterization of positive maps is still far from being thoroughly understood.

\subsection{Quantum dynamical maps and D-divisibility}

In this article we will be mainly interested in application of decomposable maps, as well as emergence of decomposability, in modeling the dynamics of finite-dimensional physical systems. Let then $\rho_t \in\matrd$ denote a time-dependent \emph{density matrix} of such a system, i.e.~$\rho_t \geqslant 0$, $\tr{\rho_t} = 1$. Recall, that quantum evolution, i.e.~a form of a function $t\mapsto\rho_t$, is described in terms of a \emph{quantum dynamical map} $\Lambda_t : \matrd\to\matrd$ such that $\rho_t = \Lambda_t (\rho_0)$ for some initial density matrix $\rho_0$. We demand $\Lambda_t$ to preserve statistical interpretation of density matrix, i.e.~we want it to be a \emph{positive map} on $\matrd$ and \emph{trace-preserving}, $\tr{\Lambda_t (\rho)} = \tr{\rho}$. We limit our analysis to dynamical maps governed by \emph{Master Equations} with time-local \emph{generators} $L_t$, i.e.~we assume $\Lambda_t$ satisfies a Cauchy problem
\begin{equation}\label{eq:ODEME}
    \frac{\der\Lambda_t}{\der t} = L_t \circ \Lambda_t, \quad \Lambda_0 = \id{}.
\end{equation}
Any such dynamical map is \emph{divisible}: for any $t\geqslant 0$ and any $s \in [0,t]$, there exists a trace-preserving linear map $V_{t,s}$, called a \emph{propagator}, s.t.
\begin{equation}
    \Lambda_t = V_{t,s}\circ\Lambda_s.
\end{equation}
Now, $V_{t,s}$ propagates a solution of the Master Equation \eqref{eq:ODEME} from an earlier time $s$ to a later time $t$. Depending on the properties of its propagator, a dynamical map $\Lambda_t$ is then commonly called
\begin{enumerate}
    \item \emph{P-divisible} if $V_{t,s}$ is a positive map, and
    \item \emph{CP-divisible} if $V_{t,s}$ is completely positive.
\end{enumerate}
The evolution provided by a CP-divisible dynamical map is called \emph{Markovian}.
In case when $\Lambda_t$ is not CP-divisible, namely if $V_{t,s}$ is not a CP map for (at least) some $(t,s)$, the evolution governed by $\Lambda_t$ is called \emph{non-Markovian}. When $\Lambda_t$ is not CP-divisible but still satisfies \eqref{eq:ODEME} for some time-local generator $L_t$, it is sometimes called \emph{weakly non-Markovian} \cite{Sabrina}. Clearly, all CP-divisible dynamical maps are also P-divisible. Also, the CP-divisibility of dynamical maps can be totally inscribed into an internal structure of the generator. Namely, by seminal results of Lindblad \cite{Lindblad1976}, Kossakowski, Gorini and Sudarshan \cite{Gorini1976} the evolution is Markovian if and only if $L_t$ is of the celebrated \emph{standard form}
\begin{equation}
    L_t(\rho) = -i\comm{H_t}{\rho} + \sum_n \left(X_{n,t} \rho X_{n,t}^{\hadj} - \frac{1}{2}\acomm{X_{n,t}^{\hadj}X_{n,t}}{\rho}\right)
\end{equation}
for some Hermitian matrix $H_t$ and a family of matrices $\{X_{n,t}\}_n$ (usually non-unique), or, equivalently,
\begin{equation}\label{eq:LtCPdivisible}
    L_t = -i\comm{H_t}{\cdot} + \phi_t - \frac{1}{2}\acomm{\phi_{t}^{\dual}(\mathbb{I})}{\cdot}
\end{equation}
for Hermitian $H_t$ and some completely positive map $\phi_t$ (also non-unique); here, $\mathbb{I}$ is an identity matrix in $\matrd$ and $\phi^\dual$ denotes a map \emph{dual} to $\phi$, i.e.~satisfying property $\tr{a\phi(b)} = \tr{\phi^\dual (a)b}$ for every $a,b\in\matrd$.
\vskip\baselineskip
Recently, a new subclass of P-divisible dynamical maps was introduced in \cite{Szczygielski2023}:
\begin{definition}
    A dynamical map $\Lambda_t$ is called decomposably divisible or D-divisible iff it is divisible and its propagator $V_{t,s}$ is a trace-preserving and decomposable map on $\matrd$ in the sense of definition \ref{def:DecomposableMap}.
\end{definition}
Since every CP map is decomposable, D-divisibility is a straightforward yet nontrivial generalization of CP-divisibility in the sense that every CP-divisible dynamical map is also D-divisible. The following necessary and sufficient condition for D-divisibility was found in \cite{Szczygielski2023}:
\begin{theorem}[cf.~\cite{Szczygielski2023}]
    A differentiable dynamical map $\Lambda_t$ satisfying Master Equation \eqref{eq:ODEME} with a time-local $L_t$ is D-divisible iff there exists a Hermitian matrix $H_t$ and a decomposable map $\varphi_t$ s.t.
    \begin{equation}\label{eq:LtDdivisible}
        L_t = -i \comm{H_t}{\cdot} + \varphi_t - \frac{1}{2}\acomm{\varphi_{t}^{\dual}(\mathbb{I})}{\cdot} .
    \end{equation}
\end{theorem}
Clearly, the family of dynamical maps governed by \eqref{eq:LtDdivisible} is a generalization of the CP-divisible case \eqref{eq:LtCPdivisible}.

\section{Qubit evolution -- Pauli channels}

Consider a unital qubit dynamical map represented by the Pauli channel
\begin{equation}
    \Lambda_t(\rho)=\sum_{\alpha=0}^3p_\alpha(t)\sigma_\alpha \rho \sigma_\alpha,
\end{equation}
where $\sigma_0=\mathbb{I}$ and $\sigma_1$, $\sigma_2$, $\sigma_3$ are the usual Pauli matrices. Note that the entire time-dependence is encoded into the probability distribution $p_\alpha(t)$. Now, if $\Lambda_t$ is the solution of master equation (\ref{eq:ODEME}) with the time-local generator
\begin{equation}\label{LPC}
    L_t=\sum_{\alpha=1}^3\gamma_\alpha(t)\mathcal{L}_\alpha,
    \qquad \mathcal{L}_\alpha(\rho)=\frac 12(\sigma_\alpha\rho\sigma_\alpha-\rho).
\end{equation}
then there is the following correspondence between $p_\alpha(t)$ and the decoherence rates $\gamma_\alpha(t)$ \cite{Filip,Filip2},
\begin{equation}
  p_0(t) = \frac 14 \left[1 + \sum_{k=1}^3\lambda_k(t)\right],
\end{equation}
\begin{equation}\
  p_k(t) = \frac 14 \left[1 + 2\lambda_k(t) - \sum_{\ell=1}^3\lambda_\ell(t)\right],\qquad k=1,2,3.
\end{equation}
In the above equations,
\begin{equation}
\lambda_k(t)=\exp\left[-\int_0^t(\gamma_0(\tau)-\gamma_k(\tau))\der\tau\right]
\end{equation}
are the eigenvalues of $\Lambda_t$ to the eigenvectors $\sigma_k$, and $\gamma_0(t)=\gamma_1(t)+\gamma_2(t)+\gamma_3(t)$. 

Divisibility properties of the Pauli dynamical maps are well characterized. An invertible $\Lambda(t)$ is CP-divisible if and only if $\gamma_\alpha(t)\geqslant0$, whereas it is only P-divisible if and only if $\gamma_\alpha(t)+\gamma_\beta(t)\geqslant0$ for all $\alpha\neq\beta$ \cite{Filip,Filip2}. Therefore, the non-Markovianity degree of the corresponding evolution is fully encoded into the decoherence rates $\gamma_\alpha(t)$. Every positive qubit map is decomposable, and hence P-divisibility is equal to D-divisibility. Actually, this weaker property can also be manifested on the level of a generator.

\begin{theorem}\label{Th_PC}
The Pauli dynamical map $\Lambda_t$ is D-divisible if and only if it is generated via a time-local generator of the form
\begin{equation}\label{form1}
L_t=\sum_{\alpha=1}^3j_\alpha(t)G_\alpha,\qquad G_\alpha=\frac 12\sum_{\beta=1}^3\mathcal{L}_\beta-\mathcal{L}_\alpha,
\end{equation}
with $j_\alpha(t)\geqslant0$.
\end{theorem}

\begin{proof}
Observe that eq. (\ref{form1}) is equivalent to
\begin{equation}
\mathcal{L}(t)=\frac 12 \sum_{\alpha=1}^3\left[\sum_{\beta=1}^3j_\beta(t)
-2j_\alpha(t)\right]\mathcal{L}_\alpha,
\end{equation}
which means that
\begin{equation}
\gamma_\alpha(t)=\frac 12 \sum_{\beta=1}^3j_\beta(t)-j_\alpha(t).
\end{equation}
The inverse relation reads
\begin{equation}\label{nc}
j_\alpha(t)=\sum_{\beta=1}^3\gamma_\beta(t)-\gamma_\alpha(t),
\end{equation}
and hence $j_\alpha(t)\geqslant0$ is necessary and sufficient for D-divisibility of the corresponding $\Lambda(t)$.
\end{proof}

An important remark is that each
\begin{equation}
G_\alpha=\varphi_\alpha-\frac 12 \{\varphi_\alpha^\prime(\mathbb{I}),\cdot\}
=\phi_\alpha-\oper
\end{equation}
is constructed from a completely copositive map 
\begin{equation}
\varphi_1(\rho)=(\sigma_3\rho\sigma_3)^\transpose,\qquad
\varphi_2(\rho)=\rho^\transpose,\qquad
\varphi_3(\rho)=(\sigma_1\rho\sigma_1)^\transpose.
\end{equation}
Therefore, Theorem \ref{Th_PC} provides an analogue of the CP-divisibility condition $\gamma_\alpha(t)\geqslant 0$ for D-divisibility. As we show in the next example, this property carries over to qudit evolution.

\section{Generalized Pauli channels}

For a qudit evolution, we consider a unital dynamical map in the form of the generalized Pauli channel \cite{Ruskai,mub_final}
\begin{equation}\label{GPC}
\Lambda_t=\frac{dp_0(t)-1}{d-1}\oper+\frac{d}{d-1}\sum_{\alpha=1}^{d+1}p_\alpha(t)\Phi_\alpha.
\end{equation}
In the above formula, $p_\alpha(t)$ are probabilities, and $\oper$ is the identity map. The $d+1$ channels
\begin{equation}
\Phi_\alpha(X)=\sum_{k=0}^{d-1}P_k^{(\alpha)}XP_k^{(\alpha)}
\end{equation}
are defined via rank-1 projectors $P_k^{(\alpha)}:=|\psi_k^{(\alpha)}\>\<\psi_k^{(\alpha)}|$ onto the vectors that form mutually unbiased bases (MUBs) $\mathcal{B}_\alpha=\{\psi_0^{(\alpha)},\dots,\psi_{d-1}^{(\alpha)}\}$. 
Recall that two bases $\mathcal{B}_\alpha$, $\mathcal{B}_\beta$ are mutually unbiased if
\begin{equation}
\big\<\psi_k^{(\alpha)}\big|\psi_l^{(\alpha)}\big\>=\delta_{kl},\qquad
\big|\big\<\psi_k^{(\alpha)}\big|\psi_l^{(\beta)}\big\>\big|^2=
\frac 1d,\quad\alpha\neq\beta.
\end{equation}
Note that the maximal number of $d+1$ MUBs is known to exist only in prime and power prime dimensions $d$ \cite{Wootters,MAX}. Therefore, the generalized Pauli channels from eq. (\ref{GPC}) can be constructed only in those dimensions.

Now, assume that the generalized Pauli dynamical map $\Lambda_t$ solves the master equation with a time-local generator
\begin{equation}\label{L_GPC}
L_t=\sum_{\alpha=1}^{d+1}\gamma_\alpha(t)\mathcal{L}_\alpha,\qquad \mathcal{L}_\alpha=\Phi_\alpha-\oper.
\end{equation}
Then, the probability distribution which defines $\Lambda_t$ is given by
\begin{align}
p_0(t)&=\frac{1}{d^2} \left[1+(d-1)\sum_{\beta=1}^{d+1}\lambda_\beta(t)\right],\\
p_\alpha(t)&=\frac{d-1}{d^2} \left[1+d\lambda_\alpha(t)-\sum_{\beta=1}^{d+1}\lambda_\beta(t)\right]
\end{align}
for $\alpha=1,\, ... \, ,d+1$. The eigenvalues $\lambda_\alpha(t)$ of the generalized Pauli map to the unitary eigenvectors
\begin{equation}\label{U}
U_\alpha^k=\sum_{\ell=0}^{d-1}\omega^{k\ell}P_\ell^{(\alpha)},\qquad \omega=e^{2\pi i/d},
\end{equation}
are $(d-1)$-times degenerated, real-valued functions. They are related to the decoherence rates $\gamma_\alpha(t)$ via
\begin{equation}
\lambda_\alpha(t)=\exp\left[-\int_0^t\left(\gamma_0(\tau)-\gamma_\alpha(\tau)\right)\der\tau\right],
\end{equation}
where $\gamma_0(t)=\sum_{\beta=1}^{d+1}\gamma_\beta(t)$. Recall that the positivity of the decoherence rates $\gamma_\alpha(t)\geqslant0$ at all times $t\geqslant0$ is equivalent to Markovianity of the corresponding evolution. However, for the qudit scenario, there are no general necessary and sufficient conditions where it comes to D-divisibility of $\Lambda_t$.

\begin{theorem}\label{thm:GeneralizedPauli}
The generalized Pauli dynamical map $\Lambda_t$ generated by $L_t$ of form \eqref{L_GPC} is D-divisible if
   \begin{equation}\label{eq:GPCbetaCondition}
       \sum_{\beta=1}^{d+1}\gamma_\beta(t)\geqslant (d-1)\max_\alpha\gamma_\alpha(t),
   \end{equation}
where the maximization is performed over all $\alpha=1,\, ... \, ,d+1$.
\end{theorem}

\begin{proof}
In ref. \cite{SICMUB_channels}, it is proved that the generalized Pauli dynamical maps generated by
\begin{equation}
G_\alpha=\gamma(t)\mathcal{L}_\alpha+\widetilde{\gamma}(t)\sum_{\beta\neq\alpha}
\mathcal{L}_\beta
\end{equation}
are D-divisible if $\widetilde{\gamma}(t)\geqslant0$ and $\gamma(t)+\widetilde{\gamma}(t)\geqslant0$. Now, recall that if two operators are decomposable, then their convex combination is also decomposable. Hence, the time-local generator constructed as follows,
\begin{equation}
L_t=\sum_{\alpha=1}^{d+1}j_\alpha(t)G_\alpha,\qquad j_\alpha(t)\geqslant0,
\end{equation}
leads to D-divisible $\Lambda_t$. In the definition of $G_\alpha$, let us choose the borderline case with $\widetilde{\gamma}(t)=1$ and $\gamma(t)=-1$. Then, by comparing the above formula with eq. (\ref{L_GPC}), we recover the correspondence between $\gamma_\alpha(t)$ and $j_\alpha(t)$;
\begin{equation}
j_\alpha(t)=\frac 12 \left[\frac{1}{d-1}\sum_{\beta=1}^{d+1}\gamma_\beta(t)-\gamma_\alpha(t)\right].
\end{equation}
Finally, the condition for positivity of $j_\alpha(t)$ recovers the sufficient D-divisibility condition from eq. (\ref{eq:GPCbetaCondition}).
\end{proof}

Observe that, for $d=2$, eq. (\ref{eq:GPCbetaCondition}) recovers the necessary and sufficient D-divisibility conditions from eq. (\ref{nc}).

Due to the lack of general construction methods for positive maps, divisibility criteria for the qudit maps are not fully characterized. For the generalized Pauli channels, there are only known P-divisibility conditions that are either necessary \cite{mub_final}
\begin{equation}\label{Pnecc}
\sum_{\beta=1}^{d+1}\gamma_\beta(t)-\gamma_\alpha(t)\geqslant0
\end{equation}
or sufficient \cite{mub_final,ICQC}
\begin{equation}\label{Psuff}
\bigforall_{\beta\neq\alpha}\quad
\gamma_\alpha+(d-1)\gamma_\beta\geqslant0,
\end{equation}
\begin{equation}\label{PCOND}
\bigforall_{\gamma_\beta(t)\geqslant0}\quad
\gamma_\beta(t)\geqslant -\frac{d+2(k-1)}{d-2(k-1)}\min_\alpha\gamma_\alpha(t),
\end{equation}
where $k\leqslant (d+1)/2$ is the number of negative decoherence rates at time $t\geqslant0$.

\begin{example}
    Take the example from ref. \cite{SICMUB_channels}, where
    \begin{equation}
L_t=\sum_{\alpha=1}^{d+1}\gamma_\alpha(t)\mathcal{L}_\alpha
    \end{equation}
    with $\gamma_\alpha(t)=\gamma$ and $\gamma_\beta(t)=\widetilde{\gamma}$ for all $\beta\neq\alpha$. These rates satisfy
    \begin{itemize}
        \item the necessary P-divisibility condition (\ref{Pnecc}) if and only if
        \begin{equation}
            \sum_{\beta=1}^{d+1}\gamma_\beta-\gamma_\alpha\geqslant0\qquad
            \Longleftrightarrow\qquad \widetilde{\gamma}\geqslant0\qquad\wedge\qquad \gamma+(d-1)\widetilde{\gamma}\geqslant0,
        \end{equation}
        \item the sufficient P-divisibility condition (\ref{Psuff}) if and only if
        \begin{equation}
            \gamma_\alpha+(d-1)\gamma_\beta\geqslant0\qquad
            \Longleftrightarrow\qquad \gamma+(d-1)\widetilde{\gamma}\geqslant0 \qquad\wedge\qquad \widetilde{\gamma}+(d-1)\gamma\geqslant0,
        \end{equation}
        \item the more inclusive P-divisibility condition (\ref{PCOND}) if and only if $\gamma<0$, $\widetilde{\gamma}\geqslant0$, and
        \begin{equation}
            \gamma_\alpha+\gamma_\beta\geqslant0\qquad
            \Longleftrightarrow\qquad \widetilde{\gamma}\geqslant0 \qquad\wedge\qquad \gamma+\widetilde{\gamma}\geqslant0,
        \end{equation}
        \item the sufficient D-divisibility conditions (ref. \cite{SICMUB_channels}) if and only if
        \begin{equation}\label{DD}
            \widetilde{\gamma}\geqslant0 \qquad\wedge\qquad \gamma+\widetilde{\gamma}\geqslant0,
        \end{equation}            
        \item the sufficient D-divisibility conditions (\ref{eq:GPCbetaCondition}) if and only if
        \begin{equation}
            \sum_{\beta=1}^{d+1}\gamma_\beta-(d-1)\gamma_\alpha\geqslant0\qquad
            \Longleftrightarrow\qquad d\widetilde{\gamma}-(d-2)\gamma\geqslant0 \qquad\wedge\qquad \gamma+\widetilde{\gamma}\geqslant0,
        \end{equation}
        which for $\gamma<0$ and $\widetilde{\gamma}\geqslant0$ reduce to eq. (\ref{DD}).
    \end{itemize}
    Therefore, for this class of generators, the sufficient D-divisibility conditions are just as strong as the most inclusive known P-divisibility conditions, provided that there is exactly one negative decoherence rate.
\end{example}

\begin{example}
    Now, assume that we have $k$ negative decoherence rates equal to $\gamma$ and $d+1-k$ positive rates equal to $\widetilde{\gamma}$. In this case, the D-divisibility conditions from eq. (\ref{eq:GPCbetaCondition}) are equivalent to
    \begin{equation}
        (d+1-k)\widetilde{\gamma}\geqslant (d-1-k)\gamma\qquad\wedge\qquad k\gamma\geqslant (k-2)\widetilde{\gamma}.
    \end{equation}
This means that the only admissible choice is $k=1$.
Meanwhile, for P-divisibility, it is enough that
\begin{equation}
\widetilde{\gamma}\geqslant\frac{d+(k-1)}{d-(k-1)}|\gamma|,
\end{equation}
which has solutions for $1\leqslant k\leqslant (d+1)/2$.
\end{example}

\begin{example}
In $d=4$, consider a special case with two doubly degenerate decoherence rates: $\gamma_2(t)=\gamma_3(t)$ and $\gamma_4(t)=\gamma_5(t)$. Without a loss of generality, we assume that $\gamma_2(t)\geqslant\gamma_4(t)$. If there is only $k=1$ negative decoherence rate $\gamma_1(t)<0$, P-divisibility of $\Lambda_t$ is guaranteed by
\begin{equation}\label{ex1}
\gamma_2(t)\geqslant|\gamma_1(t)|.
\end{equation}
On the other hand, $\Lambda_t$ is D-divisible if
\begin{equation}
\gamma_4(t)\geqslant\frac{\gamma_2(t)+|\gamma_1(t)|}{2},
\end{equation}
which together with $\gamma_2(t)\geqslant\gamma_4(t)$ reproduces eq. (\ref{ex1}). Therefore, there are indeed less D-divisible than P-divisible $\Lambda_t$ even for a single negative decoherence rate.
\end{example}

\section{Phase-covariant channels}

Consider phase-covariant qubit maps $\Phi$, for which the following covariance property holds,
\begin{equation}\label{cov_def}
\bigforall_{X \in \mathcal{B}(\mathcal{H}),\,\phi\in\mathbb{R}}\quad\Phi(e^{-i\sigma_3\phi}X e^{i\sigma_3\phi}) = e^{-i\sigma_3\phi}\Phi(X)e^{i\sigma_3\phi},
\end{equation}
with the unitary evolution $\mathcal{U}(\theta)(\rho)=\exp(-i\sigma_3\theta)\rho\exp(i\sigma_3\theta)$
corresponding to rotations on the Bloch ball. The phase-covariant channels describe physical evolution in the presence of energy absorption, energy emission, and pure dephasing \cite{phase-cov-PRL,phase-cov,PC1,PC3}. Their most general form is $\Phi=\Lambda\mathcal{U}(\theta)$, where \cite{phase-cov,phase-cov-PRL}
\begin{equation}
\Lambda[\rho]=\frac 12 \left[(\mathbb{I}+\lambda_{\ast}\sigma_3)\Tr\rho
+\lambda_1\sigma_1\Tr(\rho\sigma_1)+\lambda_1\sigma_2\Tr(\rho\sigma_2)
+\lambda_3\sigma_3\Tr(\rho\sigma_3)\right].
\end{equation}
Whenever $\lambda_\ast\neq 0$, $\Lambda$ is non-unital. Otherwise, this is the Pauli channel with a two-times degenerate eigenvalue $\lambda_1$. The stationary state (i.e., the state preserved by $\Lambda$)
\begin{equation}
\rho_\ast=\frac 12 \left[\mathbb{I}+\frac{\lambda_\ast}{1-\lambda_3}\sigma_3\right]
\end{equation}
depends on both $\lambda_\ast$ and the eigenvalue $\lambda_3$. The necessary and sufficient complete positivity conditions are given by \cite{phase-cov}
\begin{equation}
|\lambda_3|+|\lambda_\ast|\leqslant 1,\qquad 4\lambda_1^2+\lambda_\ast^2\leqslant (1+\lambda_3)^2.
\end{equation}
Now, assume that the corresponding dynamical map $\Lambda_t$, where the entire time-dependence is encoded in $\lambda_1(t)$, $\lambda_3(t)$, and $\lambda_\ast(t)$, is the solution of
\begin{equation}\label{LPPC}
L_t=\gamma_+(t)\mathcal{L}_++\gamma_-(t)\mathcal{L}_-
+\gamma_3(t)\mathcal{L}_3
\end{equation}
where
\begin{equation}
\mathcal{L}_\pm[\rho]=\sigma_\pm \rho\sigma_\mp -\frac 12 \{\sigma_\mp\sigma_\pm,\rho\},
\qquad \mathcal{L}_3[\rho]=\frac 14(\sigma_3\rho\sigma_3-\rho),
\end{equation}
and $\sigma_\pm=(\sigma_1\pm i\sigma_2)/2$ are the raising and lowering operators. Then, one finds \cite{phase-cov}
\begin{equation}
\lambda_1(t)=\exp\left\{-\frac 12 \Big[\Gamma_+(t)+\Gamma_-(t)+\Gamma_3(t)\Big]\right\},\qquad
\lambda_3(t)=\exp\Big[-\Gamma_+(t)-\Gamma_-(t)\Big],
\end{equation}
\begin{equation}
\lambda_\ast(t)=\exp\Big[-\Gamma_+(t)-\Gamma_-(t)\Big]\int_0^t
\Big[\gamma_+(\tau)-\gamma_-(\tau)\Big]\exp\Big[\Gamma_+(\tau)+\Gamma_-(\tau)\Big]
\der\tau,
\end{equation}
with $\Gamma_\mu(t)=\int_0^t\gamma_\mu(\tau)\der\tau$. Finally, the phase-covariant dynamical map is P-divisible if and only if \cite{phase-cov}
\begin{equation}\label{PCC_Pdiv}
\gamma_\pm(t)\geqslant0,\qquad \gamma_3(t)+\sqrt{\gamma_+(t)\gamma_-(t)}\geqslant0,
\end{equation}
whereas CP-divisibility is recovered for $\gamma_\pm(t)\geqslant0$, $\gamma_3(t)\geqslant0$.

For non-unital channels, the problem of generator reparametrization becomes more involved. This is due to the conditions in eq. (\ref{PCC_Pdiv}) being non-linear in the decoherence rates. Consider the following parametrization of the generator,
\begin{equation}\label{PCC_gen}
L_t=\sum_{\alpha=1}^3\beta_\alpha(t) G_\alpha,
\end{equation}
where
\begin{equation}
    G_1=\frac 14 (\mathcal{L}_++\mathcal{L}_--2\mathcal{L}_3),\qquad
    G_2=\frac 14 (\mathcal{L}_++\mathcal{L}_3),\qquad
    G_3=\frac 14 (\mathcal{L}_-+\mathcal{L}_3),
\end{equation}
and
\begin{equation}
    \beta_1(t)=\gamma_+(t)+\gamma_-(t)-\gamma_3(t),\qquad
    \beta_2(t)=3\gamma_+(t)-\gamma_-(t)+\gamma_3(t),\qquad
    \beta_3(t)=-\gamma_+(t)+3\gamma_-(t)+\gamma_3(t).
\end{equation}
All $G_k$ are generated by completely copositive maps, whereas the maps corresponding to $G_2$ and $G_3$ are also completely positive.

\begin{theorem}
    If $\Lambda_t$ is generated via $L_t$ as in eq. (\ref{PCC_gen}) with $\beta_k(t)\geqslant0$, then it is D-divisible.
\end{theorem}

The above condition is only sufficient, being equivalent to
\begin{equation}
    \gamma_\pm(t)\geqslant0,\qquad \max\{\gamma_+(t)-3\gamma_-(t),\gamma_-(t)-3\gamma_+(t)\}
    \leqslant \gamma_3(t)\leqslant \gamma_+(t)+\gamma_-(t),
\end{equation}
Our results indicate that, for non-unital maps, the generators cannot be re-parameterized in such a way that $\beta_k(t)\geqslant0$ reproduce the necessary and sufficient conditions for D-divisibility. However, this allows for sufficient conditions that are clearly seen at the level of the generator itself.

\section{Final remarks}

We characterized families of D-divisible qudit dynamical maps characterized via $d+1$ time-dependent parameters. Conveniently, decomposable divisibility of a dynamical map is already visible on the level of the time-local generator. The presented examples support the claim that our sufficient D-divisibility conditions are more restrictive than the known P-divisibility conditions. We established problems that arise when applying our formalism to maps that are non-unital. The next step toward a further characterization of non-Markovian dynamics is replicating analogical properties for wider classes of generators. An important open question is how to efficiently deal with evolutions that go beyond bistochastic. The loss of linear dependence between the divisibility conditions and the decoherence rates requires developing alternative methods.

\section{Acknowledgments}

This research was funded in whole or in part by the National Science Centre, Poland, Grant number 2021/43/D/ST2/00102. For the purpose of Open Access, the author has applied a CC-BY public copyright license to any Author Accepted Manuscript (AAM) version arising from this submission.

\bibliographystyle{beztytulow2}
\bibliography{DecomposableDynamicsBib2}

\end{document}